\documentclass[11pt]{article}

\usepackage{fullpage}
\usepackage{amsmath}
\usepackage{amsthm}
\usepackage{amssymb}
\usepackage{amsfonts}
\usepackage{hyperref}
\usepackage{color}
\usepackage{xcolor}

\usepackage{graphicx} 
\graphicspath{ {./images/} }
\usepackage[rightcaption]{sidecap}
\usepackage{wrapfig}

\usepackage{enumitem}
\setlist{itemsep=-3pt}

\newcommand{\ignore}[1]{}

\newtheorem{theorem}{Theorem}
\newtheorem{conjecture}{Conjecture}

\newtheorem{lemma}{Lemma}
\newtheorem{fact}{Fact}
\newtheorem{definition}{Definition}
\newtheorem{claim}{Claim}

\renewcommand{\Pr}{{\bf Pr}}
\newcommand{\E}{{\bf E}}

\newcommand{\cD}{{\cal D}}
\newcommand{\cU}{{\cal U}}

\newcommand{\cA}{{\cal A}}
\newcommand{\cB}{{\cal B}}
\newcommand{\cM}{{\cal M}}

\newcommand{\cS}{{\cal S}}

\newcommand{\mwidth}{{\rm {mwidth}}}
\newcommand{\width}{{\rm {width}}}

\newcommand{\dist}{{\rm {dist}}}
\newcommand{\bj}{{\boldsymbol{j}}}
\newcommand{\bx}{\boldsymbol{x}}
\newcommand{\by}{\boldsymbol{y}}

\newcommand{\bv}{\boldsymbol{v}}
\newcommand{\bw}{\boldsymbol{w}}
\newcommand{\bz}{\boldsymbol{z}}
\newcommand{\bxi}{\boldsymbol{\xi}}
\newcommand{\dE}[1]{\underset{#1}{\mathbb{E}}}

\newcommand{\wt}{{\rm wt}}

\newcommand{\one}{{\rm one}}

\renewcommand{\exp}{{\rm exp}}
\newcommand{\opt}{{\rm opt}}
\renewcommand{\dist}{{\rm dist}}

\newcommand{\xor}{\oplus}

\newcommand{\Maj}{{\textsc {Majority}}}
\newcommand{\Yes}{{\textsc {Yes}\ }}
\newcommand{\No}{{\textsc {No}\ }}

\newcommand{\correct}{}

\newcommand{\SC}{\textsc{Set-Cover\ }}

\newcommand{\Prr}[1]{\underset{\ #1}{\Pr}}

\begin{document}

\title{Superpolynomial Lower Bounds for Learning Monotone Classes}
\author{{\bf Nader H. Bshouty}\\ Dept. of Computer Science, Technion, Haifa, Israel.\footnote{
Center for Theoretical Sciences, Guangdong Technion, (GTIIT), China. }\\ bshouty@cs.technion.ac.il}

\maketitle
\begin{abstract}
\correct
Koch, Strassle, and Tan [SODA 2023], show that, under the randomized exponential time hypothesis, there is no distribution-free PAC-learning algorithm that runs in time $n^{\tilde O(\log\log s)}$ for the classes of $n$-variable size-$s$ DNF, size-$s$ Decision Tree, and $\log s$-Junta by DNF (that returns a DNF hypothesis). Assuming a natural conjecture on the hardness of set cover, they give the lower bound $n^{\Omega(\log s)}$. This matches the best known upper bound for $n$-variable size-$s$ Decision Tree, and $\log s$-Junta.

In this paper, we give the same lower bounds for PAC-learning of $n$-variable size-$s$ Monotone DNF, size-$s$ Monotone Decision Tree, and Monotone $\log s$-Junta by~DNF. This solves the open problem proposed by Koch, Strassle, and Tan and subsumes the above results. 

The lower bound holds, even if the learner knows the distribution, can draw a sample according to the distribution in polynomial time, and can compute the target function on all the points of the support of the distribution in polynomial time.
\end{abstract}
\section{Introduction}
\correct
In the distribution-free PAC learning model~\cite{Valiant84}, the learning algorithm of a class of functions $C$ has access to an unknown target function $f\in C$ through labeled examples $(x,f(x))$ where $x$ are drawn according to an unknown but fixed probability distribution $\cD$. For a class of hypothesis $H\supseteq C$, we say that the learning algorithm $\cA$ {\it PAC-learns $C$ by $H$} in time $T$ and error $\epsilon$ if for every target $f\in C$ and distribution $\cD$, $\cA$ runs in time $T$ and outputs a hypothesis $h\in H$ which, with probability at least $2/3$, is $\epsilon$-close to $f$ with respect to $\cD$. That is, satisfies $\Pr_{\bx\sim\cD}[f(\bx)\not= h(\bx)]\le \epsilon$. 

Koch et al.,~\cite{KochST}, show that, under the randomized exponential time hypothesis (ETH), there is no PAC-learning algorithm that runs in time $n^{\tilde O(\log\log s)}$ for the classes of $n$-variable size-$s$ DNF, size-$s$ Decision Tree and $\log s$-Junta by DNF. Assuming a natural conjecture on the hardness of set cover, they give the lower bound $n^{\Omega(\log s)}$. Their lower bound holds, even if the learner knows the distribution, can draw a sample according to the distribution in polynomial time and can compute the target function on all the points of the support of the distribution in polynomial time. 

In this paper, we give the same lower bounds for PAC-learning of the classes $n$-variable size-$s$ Monotone DNF, size-$s$ Monotone Decision Tree and Monotone $\log s$-Junta by DNF. This solves the open problem proposed by Koch, Strassle, and Tan~\cite{KochST}.

\subsection{Our Results}
In this paper, we prove the following three Theorems. 

\noindent
{\bf Theorem~\ref{Th2}.} {\em Assuming randomized ETH, there is a constant $c$ such that any PAC learning algorithm for $n$-variable \textsc{Monotone} $(\log s)$-\textsc{Junta}, \mbox{\textsc{{\it size}-$s$ Monotone DT}  {\it and size}-$s$ \textsc{Monotone DNF}} by {\textsc {DNF}} with $\epsilon=1/(16n)$ must take at least 
$$n^{c\frac{\log\log s}{\log\log\log s}}$$ time.}

\noindent
{\bf Theorem~\ref{TTh2}.}
{\em Assuming a plausible conjecture on the hardness of \SC\footnote{See Conjecture~\ref{conj1}.}, there is a constant $c$ such that any PAC learning algorithm for $n$-variable \textsc{Monotone} $(\log s)$-\textsc{Junta}, {\textsc{{\it size}-$s$ Monotone DT}  {\it and size}-$s$ \textsc{Monotone DNF}} by {\textsc {DNF}} with $\epsilon=1/(16n)$ must take at least 
$$n^{c\log s}$$ time.}

\noindent
{\bf Theorem~\ref{THEND}.}
{\em Assuming randomized ETH, there is a constant $c$ such that any PAC learning algorithm for $n$-variable \textsc{Monotone} $(\log s)$-\textsc{Junta}, \mbox{\textsc{{\it size}-$s$ Monotone DT}  {\it and size}-$s$ \textsc{Monotone DNF}} by {\it size}-$s$ {\textsc {DNF}} with $\epsilon=1/(16n)$ must take at least 
$$n^{c\log s}$$ time.}

All the above lower bound holds, even if the learner knows the distribution, can draw a sample according to the distribution in polynomial time and can compute the target on all the points of the support of the distribution in polynomial time.

In the following two subsections, we give the technique used in~\cite{KochST} to prove Theorem~\ref{Th2} for 
\label{PrevT} $(\log s)$-\textsc{Junta}, and the technique we use here to extend the result to \textsc{Monotone} $(\log s)$-\textsc{Junta}. 
\correct
\subsection{Previous Technique}
In~\cite{KochST}, Koch, Strassle, and Tan show that under the randomized exponential time hypothesis, there is no PAC-learning algorithm that runs in time $n^{\tilde O(\log\log n)}$ for the class of $\log n$-Junta\footnote{$k$-Junta are Boolean functions that depend on at most $k$ variables} by DNF. The results for the other classes follow immediately from this result, since all other classes contain $\log n$-Junta. All prior works \cite{AlekhnovichBFKP08,HancockJLT96} ruled out only $poly(n)$ time algorithms.

The result in~\cite{KochST} uses the hardness result of $(k,k')$-\SC where one needs to distinguish between instances that have set cover of size at most $k$ from instances that have minimum-size set cover greater than $k'$: 
\begin{enumerate}
    \item\label{firsti} For some parameters $k$ and $k'$, assuming randomized ETH, there is a constant $\lambda<1$ such that $(k,k')$-\SC cannot be solved in time $n^{\lambda k}$. 
\end{enumerate}

First, for each set cover instance $\cS$, they identify each element in the universe with an assignment in $\{0,1\}^n$ and construct in polynomial time a target function $\Gamma^\cS:\{0,1\}^n\to \{0,1\}$ and a distribution $\cD^\cS$ that satisfies: 
\begin{enumerate}
\setcounter{enumi}{1}
\item\label{bbb} The instance $\cS$ has minimum-size set cover $\opt(\cS)$ if and only if the function $\Gamma^\cS$ is a conjunction of $\opt(\cS)$ unnegated variables\footnote{Their reduction gives a conjunction of negated variable. So here, we are referring to the dual function.} over  the distribution $\cD^\cS$.\footnote{That is, there is a term $T$ with $\opt(\cS)$ variables such that for every $x$ in the support of $\cD^\cS$, $\Gamma^\cS(x)=T(x)$.}
\end{enumerate}

For a DNF $F$ and $x\in \{0,1\}^n$, they define $\width_F(x)$ to be the size of the smallest term $T$ in $F$ that satisfies $T(x)=1$. They then show that
\begin{enumerate}
\setcounter{enumi}{2}
\item\label{II2} Any DNF $F$ with expected width $\E_{\bx\sim \cD^\cS}[\width_F(\bx)]\le \opt(\cS)/2$ is $(1/(2N))$-far from $\Gamma^\cS$ with respect to $\cD^\cS$ where $N$ is the size\footnote{$N$ is the number of sets plus the size of the universe in $\cS$.} of $\cS$. That is, $\Pr_{\bx\sim \cD^\cS}[F(\bx)\not=\Gamma^\cS(\bx)]\ge 1/(2N)$.
\end{enumerate}

They then use the following gap amplification technique. They define the function $\Gamma^\cS_{\xor\ell}:(\{0,1\}^\ell)^n\to \{0,1\}$ where for $y=(y_1,\ldots,y_n)$, $y_i=(y_{i,1},\ldots,y_{i,\ell})\in\{0,1\}^\ell$, $i\in [n]$, we have $\Gamma^\cS_{\xor\ell}(y)=\Gamma^\cS(\xor y_1,\ldots,\xor y_n)$ and $\xor y_i=y_{i,1}+\cdots+y_{i,\ell}$. They also extend the distribution $\cD^\cS$ to a distribution $\cD^\cS_{\xor\ell}$ over domain $(\{0,1\}^\ell)^n$ 
and prove that 
\begin{enumerate}
\setcounter{enumi}{3}
\item\label{It1} $\Gamma^\cS_{\xor\ell}(y)$ is a $(\opt(\cS)\ell)$-Junta over $\cD^\cS_{\xor\ell}$.
\item\label{It2} Any DNF formula $F$ with expected depth $\E_{\by\sim \cD^\cS_{\xor\ell}}[\width_F(\by)]\le \opt(\cS)\ell/4$ is $(1/(4N))$-far from $\Gamma^\cS_{\xor\ell}$ with respect to $\cD_{\xor \ell}^\cS$. 
\end{enumerate}
Item~\ref{It1} follows from the definition of $\Gamma^\cS_{\xor \ell}$ and item~\ref{bbb}. To prove Item~\ref{It2}, they show that if, to the contrary, there is a DNF $F$ of expected width at most $\opt(\cS)\ell/4$ that is $1/(4N)$-close to $\Gamma^\cS_{\xor \ell}$ with respect to $\cD_{\xor\ell}^\cS$, then there is $j\in [\ell]$ and a projection of all the variables that are not of the form $y_{i,j}$ that gives a DNF $F^*$ of expected width at most $\opt(\cS)/2$ that is $1/(2N)$-close to $\Gamma^\cS$ with respect to $\cD^\cS$. Then, by item~\ref{II2}, we get a contradiction.    

They then show that 
\begin{enumerate}
\setcounter{enumi}{5}
    \item\label{II5} Any size-$s$ DNF that is $(1/(4N))$-close to $\Gamma^\cS_{\xor \ell}$ with respect to $\cD^\cS_{\xor \ell}$ has average width $\E_{\by\sim \cD^\cS_{\xor\ell}}$ $[\width_F(\by)]\le 4\log s$.
\end{enumerate}
If $F$ is $(1/(4N))$-close to $\Gamma^\cS_{\xor\ell}$ with respect to $\cD^\cS_{\xor \ell}$, then, by items~\ref{It2} and~\ref{II5}, $4\log s\ge \E_{\by\sim \cD^\cS_{\xor\ell}}$ $[\width_F(\by)]\ge \opt(\cS)\ell/4$ and then $s\ge 2^{\opt(\cS)\ell/16}$. Therefore,
\begin{enumerate}
\setcounter{enumi}{6}
    \item\label{Fii} Any DNF of size 
less than $2^{\opt(\cS)\ell/16}$ is $(1/(4N))$-far from $\Gamma^\cS_{\xor\ell}$ with respect to $\cD_{\xor\ell}^\cS$.
\end{enumerate}

Now, let $k=\tilde O(\log\log n)$. Suppose, to the contrary, that there is a PAC-learning algorithm for $\log n$-Junta by DNF with error $\epsilon=1/(8N)$ that runs in time $t=n^{\lambda k/2}=n^{\tilde O(\log\log n)}$, where $\lambda$ is the constant in item~\ref{firsti}. Given a $(k,k')$-\SC instance, we run the learning algorithm for $\Gamma^\cS_{\xor \ell}$ for $\ell=\log n/k$. If the instance has set cover at most $k$, then by item~\ref{It1}, $\Gamma^\cS_{\xor \ell}$ is $\log n$-Junta. Then the algorithm learns the target and outputs a hypothesis that is $(1/(8N))$-close to $\Gamma_{\xor \ell}^\cS$ with respect to~$\cD^\cS_{\xor\ell}$.

On the other hand, if the instance has a minimum-size set cover of at least $k'$, then any learning algorithm that runs in time $t=n^{\lambda k/2}=n^{\tilde O(\log\log n)}$ cannot output a DNF of size more than $t$ terms. By item~\ref{Fii}, any DNF of size 
less than $2^{k'\log n/(16k)}\le 2^{\opt(\cS)\ell/16}$ is $(1/(4N))$-far from $\Gamma^\cS_{\xor\ell}$ with respect to $\cD^\cS_{\xor \ell}$. By choosing the right parameters  $k$ and $k'$, we have $2^{k'\log n/(16k)}>t$, and therefore, any DNF that the algorithm outputs has error of at least $1/(4N)$. 

Therefore, by estimating the distance of the output of the learning algorithm from $\Gamma_{\xor \ell}^\cS$ with respect to $\cD_{\xor\ell}^\cS$, we can distinguish between instances that have set cover of size less than or equal to $k$ from instances that have a minimum-size set cover greater than $k'$ in time $t=n^{\lambda k/2}$. Thus, we got an algorithm for $(k,k')$-\SC that runs in time $n^{\lambda k/2}<n^{\lambda k}$. This contradicts item~\ref{firsti} and finishes the proof of the first lower bound. 

Assuming a natural conjecture on the hardness of set cover, they give the lower bound $n^{\Omega(\log s)}$. We will discuss this in Section~\ref{LLast}. 

\subsection{Our Technique}
\correct

In this paper, we also use the hardness result of $(k,k')$-\SC. As in~\cite{KochST}, we identify each element in the universe with an assignment in $\{0,1\}^n$ and use the function $\Gamma^\cS$ and the distribution $\cD^\cS$ that satisfies: 
\begin{enumerate}
\item\label{ldl} The instance $\cS$ has minimum-size set cover $\opt(\cS)$ if and only if the function $\Gamma^\cS$ is a conjunction of $\opt(\cS)$ variables over  the distribution $\cD^\cS$.
\end{enumerate} 

We then build a monotone target function $\Gamma^\cS_\ell$ and use a different approach to show that any DNF of size
less than $2^{\opt(\cS)\ell/20}$ is $(1/(8N)-2^{-\opt(\cS)\ell/20})$-far from $\Gamma^\cS_{\ell}$ with respect to $\cD^\cS_\ell$. 

We define, for any odd $\ell$, the monotone function $\Gamma^\cS_\ell:(\{0,1\}^\ell)^n\to \{0,1\}$ where for $y=(y_1,\ldots,y_n)$, $y_i=(y_{i,1},\ldots,y_{i,\ell})$ we have $\Gamma^\cS_{\ell}(y)=\Gamma^\cS(\Maj (y_1),\ldots,\Maj (y_n))$ where $\Maj$ is the majority function. A distribution $\cD^\cS_\ell$ is also defined such that
\begin{enumerate}
\setcounter{enumi}{1}
\item\label{halfl} $\Pr_{\by\sim\cD_\ell^\cS}[\Gamma_\ell^\cS(\by)=0]=\Pr_{\by\sim\cD_\ell^\cS}[\Gamma_\ell^\cS(\by)=1]=1/2$.
\end{enumerate} 
It is clear from the definition of $\Gamma^\cS_\ell$ and item~\ref{ldl} that
\begin{enumerate}
\setcounter{enumi}{2}
\item $\Gamma^\cS_{\ell}(y)$ is a monotone $(\opt(\cS)\ell)$-Junta over $\cD^\cS_{\ell}$.
\end{enumerate}
We then define the {\it monotone size} of a term $T$ to be the number of unnegated variables that appear in $T$. We first show that  
\begin{enumerate}
\setcounter{enumi}{3}
    \item\label{jh} For every DNF $F:(\{0,1\}^\ell)^n\to \{0,1\}$ of size $|F|\le 2^{\opt(\cS)\ell/5}$ that is $\epsilon$-far from $\Gamma^\cS_\ell$ with respect to $\cD_\ell^\cS$, there is another DNF $F'$ of size  $|F'|\le 2^{\opt(\cS)\ell/5}$ with terms of monotone size at most $\opt(\cS)\ell/5$ that is $(\epsilon-2^{-\opt(\cS)\ell/20})$-far from $\Gamma_\ell^\cS$ with respect to $\cD_\ell^\cS$.
\end{enumerate}
This is done by simply showing that terms of large monotone size in the DNF $F$ have a small weight according to the distribution $\cD_\ell^\cS$ and, therefore, can be removed from $F$ with the cost of $-2^{-\opt(\cS)\ell/20}$ in the error. 

We then, roughly speaking, show that 
\begin{enumerate}
\setcounter{enumi}{4}
\item\label{klk} Let $F'$ be a DNF of size $|F'|\le 2^{\opt(\cS)\ell/5}$ with terms of monotone size at most $\opt(\cS)\ell/5$. For every $y\in (\{0,1\}^\ell)^n$ in the support of $\cD_\ell^\cS$ that satisfies $\Gamma_\ell^\cS(y)=1$, either \begin{itemize}
    \item $F'(y)=0$ or
    \item $F'(y)=1$, and at least $1/(2N)$ fraction of the points $z$ below $y$ in the lattice $(\{0,1\}^\ell)^n$ that are in the support of $\cD_\ell^\cS$ satisfies $F'(z)=1$ and $\Gamma_\ell^\cS(z)=0$. 
\end{itemize} 
\end{enumerate}
By item~\ref{klk}, either $1/(4N)$ fraction of the vectors $y$ that satisfy $\Gamma_\ell^\cS(y)=1$ satisfy $F'(y)=0$ or $(1-1/(4N))/(2N)>1/(4N)$ fraction of the points $z$ that satisfy $\Gamma_\ell^\cS(z)=0$ satisfy $F'(z)=1$. Therefore, with item~\ref{halfl}, we get that $F'$ is $1/(8N)$-far from $\Gamma^\cS_\ell$ with respect to $\cD_\ell^\cS$. This, with item~\ref{jh}, implies that
\begin{enumerate}
\setcounter{enumi}{5}
    \item If $F:(\{0,1\}^\ell)^n\to \{0,1\}$ is a DNF of size $|F|<2^{\opt(\cS)\ell/20}$, then $F$ is  $(1/(8N)-2^{-\opt(\cS)\ell/20})$-far from $\Gamma^\cS_\ell$ with respect to $\cD^\cS_\ell$.
\end{enumerate}
The rest of the proof is almost the same as in~\cite{KochST}. See the discussion in subsection~\ref{PrevT} after item~\ref{Fii}.
\subsection{Upper Bounds} 
\correct

The only known distribution-free algorithm for $\log s$-Junta is the trivial algorithm that, for every set of $m=\log s$ variables $S=\{x_{i_1},\ldots,x_{i_m}\}$, checks if there is a function that depends on $S$ and is consistent with the examples. This algorithm takes $n^{O(\log s)}$ time. 

For size-$s$ decision tree and monotone size-$s$ decision tree, the classic result of Ehrenfeucht and Haussler~\cite{EhrenfeuchtH89} gives a distribution-free time algorithm that runs in time $n^{O(\log s)}$ and outputs a decision tree of size $n^{O(\log s)}$. 

The learning algorithm is as follows: Let $T$ be the target decision tree of size $s$. First, the algorithm guesses the variable at the root of the tree $T$ and then guesses which subtree of the root has size at most $s/2$. Then, it recursively constructs the tree of size $s/2$. When it succeeds, it continues to construct the other subtree. 

For size-$s$ DNF and monotone size-$s$ DNF, Hellerstein et al.~\cite{HellersteinKSS12} gave a distribution-free proper learning algorithm that runs in time $2^{\tilde O(\sqrt{n})}$.  

To the best of our knowledge, all the other results in the literature for learning the above classes are either restricted to the uniform distribution or use, in addition, a black box queries or returns hypotheses that are not DNF. 

\section{Definitions and Preliminaries} 
In this section, we give the definitions and preliminary results that are needed to prove our results. 
\subsection{Set Cover}
\correct

Let $\cS=(S,U,E)$ be a bipartite graph on $N=n+|U|$ vertices where $S=[n]$, and for every $u\in U$, $\deg(u)>0$. We say that $C\subseteq S$ is a set cover of $\cS$ if every vertex in $U$ is adjacent to some vertex in $C$. The \SC problem is to find a minimum-size set cover. We denote by $\opt(\cS)$ the size of a minimum-size set cover for $\cS$. 

We identify each element $u\in U$ with the vector $(u_1,\ldots,u_n)\in\{0,1\}^n$ where $u_i=0$ if and only if $(i,u)\in E$. We will assume that those vectors are distinct. If there are two distinct elements $u,u'\in U$ that have the same vector, then you can remove one of them from the graph. This is because every set cover that covers one of them covers the other. 

\begin{definition}
The $(k,k')$-\SC problem is the following: Given as input a set cover instance $\cS=(S,U,E)$, and parameters $k$ and $k'$. Output \Yes if $\opt(\cS)\le k$ and \No if $\opt(\cS)>k'$.
\end{definition}

\subsection{Hardness of \SC}
\correct

Our results are conditioned on the following randomized exponential time hypothesis (ETH)

\noindent
{\bf Hypothesis:}~\cite{CalabroIKP08,DellHMTW14,ImpagliazzoP01,ImpagliazzoPZ01,Tovey84}. There exists a constant $c\in (0,1)$ such that $3$-SAT on $n$ variables cannot be solved by a randomized algorithm in $O(2^{cn})$ time with success probability at least $2/3$.

The following is proved in ~\cite{Lin19}. See also Theorem~7 in \cite{KochST}
\begin{lemma}\label{Lin}~\cite{Lin19}. Let $k\le \frac{1}{2}\frac{\log\log N}{\log\log\log N}$ and $k'=\frac{1}{2}\left(\frac{\log N}{\log\log N}\right)^{1/k}$ be two integers. Assuming randomized ETH, there is a constant $\lambda\in(0,1)$ such that there is no randomized $N^{\lambda k}$ time algorithm that can solve $\left(k,k'\right)$-\SC on $N$ vertices with high probability. 
\end{lemma}

\subsection{Concept Classes}
\correct

For the lattice $\{0,1\}^n$, and $x,y\in\{0,1\}^n$, we define the partial order $x\le y$ if $x_i\le y_i$ for every~$i$. When $x\le y$ and $x\not=y$, we write $x<y$. If $x<y$, we say that $x$ is {\it below} $y$, or $y$ is {\it above} $x$. A Boolean function $f:\{0,1\}^n\to \{0,1\}$ is {\it monotone} if, for every $x\le y$, we have $f(x)\le f(y)$. A {\it literal} is a variable or negated variable. A {\it term} is a conjunction ($\wedge$) of literals. A {\it clause} is a disjunction ($\vee$) of literals. A {\it monotone term} (resp. clause) is a conjunction (resp. disjunction) of unnegated variables. The {\it size} of a term $T$, $|T|$, is the number of literals in the term $T$. A DNF (resp. CNF) is a disjunction (resp. conjunction) of terms (resp. clauses). The {\it size} $|F|$ of a DNF (resp. CNF) $F$ is the number of terms (resp. clauses) in $F$. A {\it monotone DNF} (resp. monotone CNF) is a DNF (resp. CNF) with monotone terms (resp. clauses). 

We define the following classes
\begin{enumerate}
    \item size-$s$ DNF and size-$s$ \textsc{Monotone DNF} are the classes of DNF and monotone DNF, respectively, of size at most $s$.
    \item size $s$-\textsc{DT} and size-$s$ \textsc{Monotone DT} are the classes of decision trees and monotone decision trees, respectively, with at most $s$ leaves. 
    \item $k$-\textsc{Junta} and \textsc{Monotone} $k$-\textsc{Junta} are the classes of Boolean functions and monotone Boolean functions that depend on at most $k$ variables. 
\end{enumerate}
It is well known that 
\begin{eqnarray}\label{Mon}
\mbox{\textsc{Monotone ($\log s$)-Junta$\subset$  {\rm size}-$s$ Monotone DT }$\subset$  {\rm size}-$s$ \textsc{Monotone} DNF }.
\end{eqnarray}

\subsection{Functions and Distributions}
\correct

For any set $R$, we define $\cU(R)$ to be the uniform distribution over $R$. For a distribution~$\cD$ over $\{0,1\}^n$ and two Boolean functions $f$ and $g$, we define $\dist_\cD(f,g)=\Pr_{\bx\sim \cD}[f(\bx)\not=g(\bx)]$. Here, bold letters denote random variables. If $\dist_\cD(f,g)=0$, then we say that $f=g$ {\it over $\cD$}. 
For a class of functions $C$, we say that $f$ is $C$ over $\cD$ if there is a function $g\in C$ such that $f=g$ over $\cD$.
\begin{definition}
($\Gamma^{\cS}$ and $\cD^{\cS}$)
Let $\cS=(S,U,E)$ be a set cover instance with $S=[n]$. Recall that we identify each element $u\in U$ with the vector $(u_1,\ldots,u_n)\in\{0,1\}^n$ where $u_i=0$ if and only if $(i,u)\in E$. 
We define the partial function $\Gamma^\cS:\{0,1\}^n\to\{0,1\}$ where $\Gamma^\cS(x)=0$ if $x\in U$ and $\Gamma^\cS(1^n)=1$. We define the distribution $\cD^\cS$ over $\{0,1\}^n$ where $\cD^\cS(x)=1/2$ if $x=1^n$, $\cD^\cS(x)=1/(2|U|)$ if $x\in U$, and $\cD^\cS(x)=0$ otherwise. We will remove the superscript  $\cS$ when it is clear from the context and write $\Gamma$ and $\cD$.
\end{definition}

\begin{fact}\label{factopt} We have
\begin{enumerate}
    \item $C\subseteq S$ is a set cover of $\cS=(S,U,E)$, if and only if $\Gamma(x)=\bigwedge_{i\in C} x_i$ over $\cD$. 
    \item In particular, If $T$ is a monotone term of size $|T|<\opt(\cS)$, then there is $u\in U$ such that~$T(u)=1$.
\end{enumerate}
\end{fact}
\begin{proof}
Let $C$ be a set cover of $\cS$. First, we have $\Gamma(1^n)=1$. Now, since $C$ is a set cover, every vertex $u\in U$ is adjacent to some vertex in $C$. This is equivalent to: for every assignment $u\in U$, there is $i\in C$ such that $u_i=0$. Therefore, $\wedge_{i\in C} u_i=0$ for all $u\in U$. Thus, $\Gamma(x)=\bigwedge_{i\in C} x_i$ over~$\cD$.

The other direction can be easily seen by tracing backward in the above proof.
\end{proof}
For an odd $\ell$, define $\Delta^0=\{a\in \{0,1\}^\ell|\wt(a)=\lfloor\ell/2\rfloor\}$ and $\Delta^1=\{a\in \{0,1\}^\ell|\wt(a)= \lceil\ell/2\rceil\}$, where $\wt(a)$ is the Hamming weight of $a$. 
Notice that $|\Delta^0|=|\Delta^1|={\ell\choose \lfloor\ell/2\rfloor}$. 

\begin{definition} ($\Gamma_\ell$, $\cD_\ell$, $\Delta^0_n$ and $\Delta^1_n$) For an odd $\ell$, define $\Delta^1_n=(\Delta^1)^n$ and\footnote{Here $\Delta^\xi=\Delta^0$ if $\xi=0$ and $\Delta^1$ if $\xi=1$.} $\Delta^0_n:=\cup_{u\in U}\prod_{i=1}^n\Delta^{u_i}=\cup_{u\in U}(\Delta^{u_1}\times \Delta^{u_2}\times \cdots\times \Delta^{u_n})$. 
Define the distribution $\cD_\ell:(\{0,1\}^\ell)^n\to [0,1]$ to be $\cD_\ell(y)=1/(2|\Delta^1_n|)=1/(2|\Delta^1|^n)$ if $y\in \Delta^1_n$, $\cD_\ell(y)=1/(2|\Delta^0_n|)=1/(2|U|\cdot|\Delta^0|^n)$ if $y\in \Delta^0_n$, and $\cD_\ell(y)=0$ otherwise. 
We define the partial function $\Gamma_\ell$ over the support $\Delta^0_n\cup\Delta^1_n$ of $\cD_\ell$ to be $1$ if $y\in \Delta^1_n$ and $0$ if $y\in \Delta^0_n$.
\end{definition}
We note here that the distribution $\cD_\ell$ is well-defined. This is because: First, the sum of the distribution of the points in $\Delta_n^1$ is $1/2$. 
Second, for two different $u,u'\in U$, we have that $\prod_{i=1}^n\Delta^{u_i}$ and $\prod_{i=1}^n\Delta^{u'_i}$ are disjoint sets. Therefore, $|\Delta_n^0|=|U|\cdot |\Delta^0|^n$, and therefore, the sum of the distribution of all the points in $\Delta_n^0$ is half. In particular, 
\begin{fact}\label{halff}
We have $\Prr{\by\sim \cD_\ell}[\Gamma_\ell(\by)=1]=\Prr{\by\sim \cD_\ell}[\Gamma_\ell(\by)=0]=\Prr{\cD_\ell}[\Delta^1_n]=\Prr{\cD_\ell}[\Delta^0_n]=\frac{1}{2}.$
\end{fact}
For $y\in(\{0,1\}^\ell)^n$, we write $y=(y_1,\ldots,y_n)$, where $y_j=(y_{j,1},y_{j,2},\ldots,y_{j,\ell})\in\{0,1\}^\ell$.  Let $(\Maj(y_i))_{i\in [n]}=(\Maj(y_1),\ldots,\Maj(y_n))$ where $\Maj$ is the majority function.

\begin{fact}\label{oell}
If $C\subseteq S$ is a set cover of $\cS$, then $\Gamma_\ell(y)=\Gamma((\Maj(y_i))_{i\in [n]})=\bigwedge_{i\in C}\Maj(y_i)$ over $\cD$.  In particular, $\Gamma_\ell$ is \textsc{Monotone $\opt(\cS)\ell$-Junta} over $\cD$.
\end{fact}
\begin{proof}
First notice that $\Maj(x)=1$ if $x\in \Delta^1$ and $\Maj(x)=0$ if $x\in \Delta^0$. Therefore, for $x\in \Delta^\xi$, $\xi\in \{0,1\}$ we have $\Maj(x)=\xi$. 

For $y\in \Delta_n^1=(\Delta^1)^n$, $(\Maj(y_i))_{i\in [n]}=1^n$ and $\Gamma_\ell(y)=1=\Gamma(1^n)$.

For $y\in \Delta^0_n=\cup_{u\in U}(\Delta^{u_1}\times \Delta^{u_2}\times \cdots\times \Delta^{u_n})$, there is $u$ such that $y\in \Delta^{u_1}\times \Delta^{u_2}\times \cdots\times \Delta^{u_n}$. Then, $(\Maj(y_i))_{i\in [n]}=u$ and $\Gamma_\ell(y)=\Gamma((\Maj(y_i))_{i\in [n]})=\Gamma(u)=0$.
\end{proof}
For $t\in [\ell],\xi\in \{0,1\}$ and $u\in \{0,1\}^\ell$, we define $u^{t\gets\xi}\in \{0,1\}^\ell$ the vector that satisfies
$$u^{t\gets \xi}_i=\left\{\begin{array}{ll}
u_i&i\not=t\\
\xi&i=t
\end{array}\right..$$ 

Let $z\in (\{0,1\}^\ell)^n$. For $j\in [\ell]^n$ and $a\in \{0,1\}^n$, define $z^{j\gets a}=(z_1^{j_1\gets a_1},\ldots,z_n^{j_n\gets a_n})$. 
For a set $V\subseteq \{0,1\}^n$, we define $z^{j\gets V}=\{z^{j\gets v}|v\in V\}$.

We define $\one(z)=\prod_{i=1}^n\{m_i|z_{i,m_i}=1\}=\{m_1|z_{1,m_1}=1\}\times \cdots\times \{m_n|z_{n,m_n}=1\}$.
\begin{fact}\label{Factz}
Let $w\in \Delta_n^1$, $j\in \one(w)$, and $T$ be a term that satisfies $T(w)=1$. Then
\begin{enumerate}
\item\label{Factz1} $w^{j\gets U}\subseteq \Delta_n^0$.
    \item\label{Factz2} $|w^{j\gets U}|=|U|$.
    \item\label{Factz3} If $T^j(y_{1,j_1},\ldots,y_{n,j_n})$ is the conjunction of all the variables that appear in $T$ of the form $y_{i,j_i}$, then $T(w^{j\gets a})=T^j(a)$.
\end{enumerate}
\end{fact}

\begin{proof} We first prove item~\ref{Factz1}. Let $u\in U$ and $i$ be any integer in $[n]$. Since $w\in \Delta_n^1$, we have $w_i\in \Delta^1$. Since $j\in one(w)$, we have $w_{i,j_i}=1$. Therefore, $w_i^{j_i\gets u_i}\in \Delta^{u_i}$ for all $i\in [n]$ and $w^{j\gets u}\in \prod_{i=1}^n\Delta^{u_i}$. Thus, $w^{j\gets u}\in \Delta_n^0$ for all $u\in U$. 

To prove item~\ref{Factz2}, let $u,u'$ be two distinct elements of $U$. There is $i$ such that $u_i\not=u'_i$. Therefore $w_i^{j_i\gets u_i}\not=w_i^{j_i\gets u'_i}$ and $w^{j\gets u}\not=w^{j\gets u'}$. 

We now prove item~\ref{Factz3}. Let $T'$ be the conjunction of all the variables that appear in $T$ that are not of the form $y_{i,j_i}$. Then $T=T'\wedge T^j$. Since $T(w)=1$, we have $T'(w)=1$. Since the entries of  $w^{j\gets a}$ are equal to those in $w$ on all the variables that are not of the form $y_{i,j_i}$, we have $T'(w^{j\gets a})=1$. Therefore, $T(w^{j\gets a})=T'(w^{j\gets a})\wedge T^j(w^{j\gets a}_{1,j_1},\ldots,w^{j\gets a}_{n,j_n})=T^j(a)$.
\end{proof}
We now give a different way of sampling according to the distribution $\cD_\ell$.
\begin{fact}\label{sample} Let $\cS$ be a \SC instance.  The following is an equivalent way of sampling from $\cD_\ell$.
\begin{enumerate}
\item Draw $\xi\in\{0,1\}$ u.a.r.\footnote{Uniformly at random.}
\item Draw $w\in \Delta^1_n$ u.a.r.
\item If $\xi=1$ then output $y=w$.
\item If $\xi=0$ then 
\begin{enumerate}
    \item draw $j\in \one(w)$ u.a.r.
    \item draw $v\in w^{j\gets U}$ u.a.r.
    \item output $y=v$.
\end{enumerate}
\end{enumerate}
In particular, for any event $X$,  
$$\Prr{\by\sim\cU(\Delta_n^0)}[X]=\Prr{\bw\sim\cU(\Delta^1_n), \bj\sim\cU(\one(\bw)),\by\sim \cU(\bw^{\bj\gets U})}[X].$$
\end{fact}
\begin{proof} Denote the above distribution by $\cD'$. By Item~\ref{Factz1} in Fact~\ref{Factz}, if $w\in \Delta_n^1$ and $j\in\one(w)$, then $w^{j\gets U}\subseteq \Delta^0_n$. Therefore,  
for $z\in \Delta^1_n$, $\Pr_{\by\sim \cD'}[\by=z|\bxi=0]=0$ and then
$$\Prr{\by\sim\cD'}[\by=z]=\Prr{\bxi\sim \cU({\{0,1\})}}{[\bxi=1]}\cdot \Prr{\by\sim \cU({\Delta^1_n})}[\by=z]=\frac{1}{2|\Delta^1_n|}=\frac{1}{2|\Delta^1|^n}.$$
For $z\in\Delta^0_n$, suppose $z\in \Delta^{u_1}\times\cdots\times \Delta^{u_n}$ where $u\in U$. 
In the sampling according to $\cD'$ and when $\xi=0$, since for $j\in \one(w)$, the elements of $w^{j\gets U}$ are below $w$, we have $\Prr{\by\sim \cD'}[\by=z|\bw\not>z]=0$. Therefore,
\begin{eqnarray}\label{eq01}
\Prr{\by\sim\cD'}[\by=z]&=&\Prr{\bxi\sim\cU({\{0,1\}})}[\bxi=0]\cdot \Prr{\bw\sim\cU({\Delta^1_n})}[\bw>z]\cdot\nonumber\\
&&\cdot \Prr{\bj\sim \cU(one(\bw))}[z\in \bw^{\bj\gets U}|\bw>z,\bw\in \Delta^1_n]\cdot \Prr{\bv\sim \cU(\bw^{\bj\gets U})}[\bv=z|z\in \bw^{\bj\gets U}].
\end{eqnarray}

Now, since, for $x\in \Delta^0$, the number of elements in $\Delta^1$ that are above $x$ is $\lceil\ell/2\rceil$, we have that
the number of $w\in \Delta^1_n=(\Delta^1)^n$ that are above $z\in\Delta^{u_1}\times\cdots\times \Delta^{u_n}$ is $\lceil \ell/2\rceil^{n-\wt(u)}$. Therefore, 
\begin{eqnarray}\label{eq02}
\Prr{\bw\sim\cU({\Delta^1_n})}[\bw>z]=\frac{\lceil \ell/2\rceil^{n-\wt(u)}}{|\Delta^1_n|}.
\end{eqnarray}
Now let $w>z$ and $w\in \Delta^1_n$. Since for two different $u,u'\in U$, we have $\prod_{i=1}^n\Delta^{u_i}$ and $\prod_{i=1}^n\Delta^{u'_i}$ are disjoint sets, and since $z\in\Delta^{u_1}\times\cdots\times \Delta^{u_n}$, we have $z\in w^{j\gets U}$ if and only if $z=w^{j\gets u}$. Therefore, the number of elements $j\in \one(w)$ that satisfy $z\in w^{j\gets U}$ is the number of elements $j\in \one(w)$ that satisfy $z=w^{j\gets u}$. This is the number of elements $j\in \one(w)$ that satisfies for every $u_i=0$, $z_{i,j_i}=0$. 
For a $j$ u.a.r. and a fixed $i$ where $u_i=0$, the probability that $z_i$ and $w_i$ differ only in entry $j_i$ is $1/\lceil \ell/2\rceil$. Therefore,
\begin{eqnarray}\label{eq03}
\Prr{\bj\sim \cU(one(\bw))}[z\in \bw^{\bj\gets U}|\bw>z,\bw\in \Delta^1_n]=\frac{1}{\lceil\ell/2\rceil^{n-\wt(u)}}.
\end{eqnarray}
Finally, by item~\ref{Factz2} in Fact~\ref{Factz}, since $|w^{j\gets U}|=|U|$, we have
\begin{eqnarray}\label{eq04}
\Prr{\bv\sim \cU(\bw^{j\gets U})}[\bv=z|z\in \bw^{j\gets U}]=\frac{1}{|\bw^{j\gets U}|}=\frac{1}{|U|}.
\end{eqnarray}
By (\ref{eq01}), (\ref{eq02}), (\ref{eq03}), and (\ref{eq04}), we have
\begin{eqnarray*}
\Prr{\by\sim\cD'}[\by=z]
&=&\frac{1}{2}\cdot \frac{\lceil\ell/2\rceil^{n-\wt(u)}}{|\Delta_n^1|}\cdot \frac{1}{\lceil\ell/2\rceil^{n-\wt(u)}}\cdot \frac{1}{|U|}=\frac{1}{2|U|\cdot |\Delta_n^1|}=\frac{1}{2|U|\cdot |\Delta^0|^n}.
\end{eqnarray*}
\end{proof}

\section{Main Lemma}
\correct

In this section, we prove
\begin{lemma}\label{error}
Let $\cS=(S,U,E)$ be a set cover instance. If $F:(\{0,1\}^\ell)^n\to \{0,1\}$ is a DNF of size $|F|<2^{\opt(\cS)\ell/20}$, then $\dist_{\cD_\ell}(F,\Gamma_\ell)\ge 1/(8|U|)-2^{-\opt(\cS)\ell/20}$.
\end{lemma}
Note that Lemma~\ref{error} is used to prove Theorem~\ref{Th2} and~\ref{TTh2}. To prove Theorem~\ref{THEND}, we will need Lemma~\ref{errorA}, a stronger version of Lemma~\ref{error}. 

To prove the lemma, we first establish some results. 

For a term $T$, let $T_\cM$ be the conjunction of all the unnegated variables in $T$. We define the {\it monotone size} of $T$ to be $|T_\cM|$. 
\begin{claim}\label{Truncate}
Let $\cS=(S,U,E)$ be a set cover instance and $\ell\ge5$. If $F:(\{0,1\}^\ell)^n\to \{0,1\}$ is a DNF of size $|F|<2^{\opt(\cS)\ell/20}$, then there is a DNF, $F'$, of size $|F'|\le2^{\opt(\cS)\ell/20}$ with terms of monotone size at most $\opt(\cS)\ell/5$ such that $\dist_{\cD_\ell}(\Gamma_\ell,F')\le \dist_{\cD_\ell}(\Gamma_\ell,F)+2^{-\opt(\cS)\ell/20}$.
\end{claim}
\begin{proof} Let $T$ be a term of monotone size at least $\opt(\cS)\ell/5$. Let $b_i$ denote the number of unnegated variables of $T$ of the form $y_{i,j}$ and let $T_i$ be their conjunction. Then $T_\cM=\wedge_{i=1}^n T_i$ and $\sum_{i=1}^nb_i=|T_\cM|\ge \opt(\cS)\ell/5$. If, for some $i$, $b_i>\lceil\ell/2\rceil$, then the term $T_i$ is zero on all $\Delta^0\cup\Delta^1$, and therefore, $T$ is zero on all $\Delta^0_n\cup\Delta^1_n$. Thus, it can be just removed from $F$. So, we may assume that $b_i\le \lceil\ell/2\rceil$ for all $i$. First,
\begin{eqnarray}
\Prr{\by\sim \cD_\ell}[T(\by)=1|\Gamma_\ell(\by)=1]
&=&\Prr{\by\sim \cU({\Delta^1_n})}[T(\by)=1]\le\Prr{\by\sim \cU({\Delta^1_n})}[T_\cM(\by)=1] \nonumber\\
&=&\prod_{i=1}^n\Prr{\by_i\sim \cU({\Delta^1})}[T_i(\by_i)=1] \nonumber\\
&=&\prod_{i=1}^n\frac{{\ell-b_i\choose \lceil\ell/2\rceil-b_i}}{{\ell\choose \lceil\ell/2\rceil}}\nonumber\\
&=&\prod_{i=1}^n \left(1-\frac{b_i}{\ell}\right)\left(1-\frac{b_i}{\ell-1}\right)\cdots \left(1-\frac{b_i}{\lceil\ell/2\rceil+1}\right)\nonumber\\
&\le& \prod_{i=1}^n \prod_{j=\lceil\ell/2\rceil+1}^\ell \exp({-b_i/j})=\prod_{i=1}^n  \exp\left({-b_i\sum_{j=\lceil\ell/2\rceil+1}^\ell 1/j}\right)\nonumber\\
&=&\exp\left({-|T_\cM|\sum_{j=\lceil\ell/2\rceil+1}^\ell 1/j}\right)
\le 2^{-|T_\cM|/2}\le 2^{-\opt(\cS)\ell/10}.\label{lstep}
\end{eqnarray}
Let $F'$ be the disjunction of all the terms in $F$ of monotone size at most $\opt(\cS)\ell/5$. Let $T^{(1)},\ldots,T^{(m)}$ be all the terms of monotone size greater than $\opt(\cS)\ell/5$ in $F$. Then, by (\ref{lstep}) and the union bound,
\begin{eqnarray}\label{kji}
\Prr{\by\sim\cD_\ell}[F(\by)\not=F'(\by)|\Gamma_\ell(\by)=1]&\le& \Prr{\by\sim\cD_\ell}[\vee_{i=1}^m T^{(i)}(\by)=1|\Gamma_\ell(\by)=1]\nonumber\\ &\le& 2^{-\opt(\cS)\ell/10}m
\le 2^{-\opt(\cS)\ell/20}.
\end{eqnarray}
and (Here we abbreviate $F'(\by),F(\by)$ and $\Gamma_\ell(\by)$ by $F',F$ and $\Gamma_\ell$)
\begin{eqnarray}
\dist_{\cD_\ell}(\Gamma_\ell,F')&=& \Prr{\by\sim\cD_\ell}[F'\not=\Gamma_\ell]\nonumber\\
&=&\frac{1}{2}\Prr{\by\sim\cD_\ell}[F'\not=\Gamma_\ell|\Gamma_\ell=1]+\frac{1}{2}\Prr{\by\sim\cD_\ell}[F'\not=\Gamma_\ell|\Gamma_\ell=0]\label{kkk1}\\
&=&\frac{1}{2}\Prr{\by\sim\cD_\ell}[F'\not=F|\Gamma_\ell=1]+\nonumber\\ &&\frac{1}{2}\Prr{\by\sim\cD_\ell}[F\not=\Gamma_\ell|\Gamma_\ell=1]+\frac{1}{2}\Prr{\by\sim\cD_\ell}[F'\not=\Gamma_\ell|\Gamma_\ell=0]\label{kkk2}\\
&\le& 2^{-\opt(\cS)\ell/20}+\frac{1}{2}\Prr{\by\sim\cD_\ell}[F\not=\Gamma_\ell|\Gamma_\ell=1]+\frac{1}{2}\Prr{\by\sim\cD_\ell}[F\not=\Gamma_\ell|\Gamma_\ell=0]\label{kkk3}\\
&=&2^{-\opt(\cS)\ell/20}+\dist_{\cD_\ell}(\Gamma_\ell,F).\nonumber
\end{eqnarray}
In (\ref{kkk1}), we used Fact~\ref{halff}. In~(\ref{kkk2}), we used the probability triangle inequality. In~(\ref{kkk3}), we used (\ref{kji}) and the fact that if $F'(\by)\not=0$, then $F(\by)\not=0$. 
\end{proof}

We now prove
\begin{claim}\label{coor}
Let $z\in \Delta_n^1$. Let $F$ be a DNF with terms of monotone size at most $\lceil\ell/2\rceil(\opt(\cS)-1)/2$ that satisfies $F(z)=1$. Then 
$$\Prr{\bj\sim \cU(\one(z)),\by\sim\cU(z^{\bj\gets U})} [F(\by)=1]\ge \frac{1}{2|U|}.$$
\end{claim}
\begin{proof} Since $F(z)=1$, there is a term $T$ in $F$ that satisfies $T(z)=1$. Let $Y_0=\{y_{i,m}|z_{i,m}=0\}$ and $Y_1=\{y_{i,m}|z_{i,m}=1\}$. Since $T(z)=1$, every variable in $Y_0$ that appears in $T$ must be negated, and every variable in $Y_1$ that appears in $T$ must be unnegated. 
For $j\in \one(z)$, define $q(j)$ to be the number of variables in $\{y_{1,j_1},\ldots,y_{n,j_n}\}$ that appear in $T(y)$. All those variables appear unnegated in $T$ because $j\in \one(z)$. Recall that $T_\cM$ is the conjunction of all unnegated variables in $T$. Then $|T_\cM|\le \lceil\ell/2\rceil(\opt(\cS)-1)/2$. Each variable in $T_\cM$ contributes $\lceil\ell/2\rceil^{n-1}$ to the sum $\sum_{j\in one(z)}q(j)$ and $|\one(z)|=\lceil\ell/2\rceil^n$. Therefore,
$$\dE{\bj\sim \cU(\one(z))}[q(\bj)]=\frac{|T_\cM|}{\lceil \ell/2\rceil}\le\frac{\opt(\cS)-1}{2}.$$
By Markov's bound, at least half the elements $j\in \one(z)$ satisfies $q(j)\le \opt(\cS)-1$. Let $J=\{j\in\one(z)|q(j)<\opt(\cS)\}$. Then $\Pr_{\bj\sim\cU(\one(z))}[\bj\in J]\ge 1/2$. Consider $j\in J$ and let $T^j$ be the conjunction of all the variables that appear in $T$ of the form $y_{i,j_i}$. Then $|T^j|=q(j)\le \opt(\cS)-1$. 
By Fact~\ref{factopt}, there is $u\in U$ such that $T^j(u)=1$.
By Fact~\ref{Factz}, we have  $T(z^{j\gets u})=T^j(u)=1$. Then $F(z^{j\gets u})=1$. Since by item~\ref{Factz1} in Fact~\ref{Factz}, $|z^{j\gets U}|=|U|$, we have $$\Pr_{\bj\sim\cU(one(z)),\by\sim\cU(z^{\bj\gets U})}[F(\by)=1|\bj\in J]\ge \frac{1}{|U|}.$$ 
Therefore,
$$\Prr{\bj\sim \cU(\one(z)),\by\sim\cU(z^{\bj\gets U})} [F(\by)=1]\ge  \Prr{\bj\sim\cU(\one(z))}[\bj\in J]\cdot \Prr{\bj\sim\cU(\one(z)),\by\sim\cU(z^{j\gets U})}[F(\by)=1|\bj\in J]\ge \frac{1}{2|U|}.$$
\end{proof}
We are now ready to prove

\addtocounter{lemma}{-1}
\begin{lemma}
Let $\cS=(S,U,E)$ be a set cover instance, and let $\ell\ge 5$. If $F:(\{0,1\}^\ell)^n\to \{0,1\}$ is a DNF of size $|F|<2^{\opt(\cS)\ell/20}$, then $\dist_{\cD_\ell}(F,\Gamma_\ell)\ge 1/(8|U|)-2^{-\opt(\cS)\ell/20}$.
\end{lemma}
\begin{proof} By Claim~\ref{Truncate}, there is a DNF, $F'$, of size $|F'|\le2^{\opt(\cS)\ell/20}$ with terms of monotone size at most $\opt(\cS)\ell/5$ such that $\dist_{\cD_\ell}(\Gamma_\ell,F')\le \dist_{\cD_\ell}(\Gamma_\ell,F)+2^{-\opt(\cS)\ell/20}$. Therefore, it is enough to prove that $\dist_{\cD_\ell}(\Gamma_\ell,F')\ge 1/(8|U|)$.

If $\Prr{\by\sim\cU(\Delta_n^1)}[F'(\by)\not=1]\ge 1/(4|U|)$, then by Fact~\ref{halff}, we have 
$$\dist_{\cD_\ell}(\Gamma_\ell,F')\ge \Prr{\by\sim\cD_\ell}[\Gamma_\ell(\by)\not=F'(\by)|\Gamma_\ell(\by)=1]\Prr{\by\sim \cD_\ell}[\Gamma_\ell(\by)=1]=\frac{1}{2}\Prr{\by\sim \cU(\Delta^1_n)}[F'(\by)\not=1]\ge \frac{1}{8|U|}.$$
If $\Prr{\by\sim\cU(\Delta_n^1)}[F'(\by)\not=1]< 1/(4|U|)$, then by Fact~\ref{halff} and~\ref{sample}, and Claim~\ref{coor}, 
\begin{eqnarray*}
\dist_{\cD_\ell}(\Gamma_\ell,F')&\ge& \Prr{\by\sim\cD_\ell}[\Gamma_\ell(\by)\not=F'(\by)|\Gamma_\ell(\by)=0]\Prr{\by\sim\cD_\ell}[\Gamma_\ell(\by)=0]\\
&=&\frac{1}{2}\Prr{\by\sim\cU(\Delta_n^0)}[F'(\by)=1]\\
&=&\frac{1}{2}\Prr{\bz\sim\cU(\Delta^1_n), \bj\sim\cU(\one(\bz)),\by\sim \cU(\bz^{\bj\gets U})}[F'(\by)=1]\\
&\ge&\frac{1}{2}\Prr{\bz\sim\cU(\Delta^1_n), \bj\sim\cU(\one(\bz)),\by\sim \cU(\bz^{\bj\gets U})}[F'(\by)=1|F'(\bz)=1]\cdot \Prr{\bz\sim\cU(\Delta^1_n)}[F'(\bz)=1]\\
&\ge&\frac{1}{2}\frac{1}{2|U|}\left(1-\frac{1}{4|U|}\right)\ge \frac{1}{8|U|}.
\end{eqnarray*}

\end{proof}

\section{Superpolynomial Lower Bound}
\correct

In this section, we prove the first results of the paper. First, we prove the following result for {{\textsc {Monotone $(\log n)$-Junta}}}.
\begin{lemma}\label{Th01}
Assuming randomized ETH, there is a constant $c$ such that any PAC learning algorithm for $n$-variable {{\textsc {Monotone $(\log n)$-Junta}}} by {\textsc {DNF}} with $\epsilon=1/(16n)$ must take at least 
$$n^{c\frac{\log\log n}{\log\log\log n}}$$ time.

The lower bound holds, even if the learner knows the distribution, can draw a sample according to the distribution in polynomial time and can compute the target on all the points of the support of the distribution in polynomial time.
\end{lemma} 
\begin{proof} Consider the constant $\lambda$ in Lemma~\ref{Lin}. Let $c=\min(1/40,\lambda/4)$.
Suppose there is a PAC learning algorithm $\cA$ for {{\textsc {Monotone $(\log n)$-Junta}}} by {\textsc {DNF}} with $\epsilon=1/(16n)$ that runs in time $n^{c\frac{\log\log n}{\log\log\log n}}$. We show that there is $k$ such that for
$$k'=\frac{1}{2}\left(\frac{\log N}{\log\log N}\right)^{1/k},$$ $(k,k')$-\SC
can be solved in time $N^{4ck}\le N^{\lambda k}$. By Lemma~\ref{Lin}, the result then follows.

Let $\cS=(S,U,E)$ be an $N$-vertex $(k,k')$-\SC instance where 
$$k=\frac{1}{2}\frac{\log\log N}{\log\log\log N} \mbox{\ and\ } k'=\frac{1}{2}\left(\frac{\log N}{\log\log N}\right)^{1/k}.$$
Let $$\ell=\frac{\log N}{k}$$ and consider $\Gamma_\ell$ and $\cD_\ell$. 

Consider the following algorithm $\cB$
\begin{enumerate}
\item Input $\cS=(S,U,E)$ an instance for $(k,k')$-\SC.
\item Construct $\Gamma_\ell$ and $\cD_\ell$.
\item Run $\cA$ using $\Gamma_\ell$ and $\cD_\ell$.  If it runs more than $N^{4c k}$ steps, then output \No.
\item Let $F$ be the output DNF.
\item Estimate $\eta=\dist_{\cD_\ell}(F,\Gamma_\ell)$.
\item If $\eta\le \frac{1}{16N}$, output \Yes, otherwise output \No.
\end{enumerate}
The running time of this algorithm is $N^{4c k}\le N^{\lambda k}$. Therefore, it is enough to prove the following
\begin{claim}
Algorithm $\cB$ solves $(k,k')$-\SC.
\end{claim}
\begin{proof}
\Yes case: Let $\cS=(S,U,E)$ be a $(k,k')$-\SC instance and $\opt(\cS)\le k$. Then, $\opt(\cS)\cdot \ell\le k\ell=\log N$, and by Fact~\ref{oell}, $\Gamma_\ell$ is \textsc{Monotone $\log N$-Junta}. Therefore, w.h.p., algorithm $\cA$ learns $\Gamma_\ell$ and outputs a DNF that is $\eta=1/(16N)$ close to the target with respect to $\cD_\ell$. Since $\cB$ terminates $\cA$ after $N^{4ck}$ time, we only need to prove that $\cA$ runs at most $N^{4c k}$ time.

The running time of $\cA$ is
$$N^{c\frac{\log\log N}{\log\log\log N}}<N^{4ck}.$$

\No Case: Let $\cS=(S,U,E)$ be a $(k,k')$-\SC instance and $\opt(\cS)> k'$. By Lemma~\ref{error}, any DNF, $F$, of size $|F|<2^{k'\ell/20}$ satisfies $\dist_{\cD_\ell}(F,\Gamma_\ell)\ge 1/(8|U|)-2^{-k'\ell/20}$.
First, we have
$$(2k)^{2k}=\left(\frac{\log\log N}{\log\log\log N}\right)^\frac{\log\log N}{\log\log\log N}<\frac{\log N}{\log\log N}.$$
Therefore, since $c\le 1/40$,
$$k'=\frac{1}{2}\left(\frac{\log N}{\log\log N}\right)^{1/k}>\frac{1}{2}(2k)^2>80c k^2.$$ 
So $k'\ell/20>(k\ell)(4ck)$ and
$$2^{k'\ell/20}>(2^{k\ell})^{4c k}=N^{4c k}.$$
Now since the algorithm runs in time $N^{4ck}$, it cannot output a DNF $F$ of size more than $N^{4ck}<2^{k'\ell/20}$, and by Lemma~\ref{error}, 
$$\dist_{\cD_\ell}(F,\Gamma_\ell)\ge \frac{1}{8|U|}-\frac{1}{N^{4ck}}\ge\frac{1}{9N}.$$
So it either runs more than $N^{4ck}$ steps and then outputs \No in step 3 or outputs a DNF with an error greater than $1/(9N)>1/(16N)$ and outputs \No in step 6.
\end{proof}
Notice that the learning algorithm knows $\Gamma_\ell$ and $\cD_\ell$. It is also clear from the definition of $\Gamma_\ell$ and $\cD_\ell$ that the learning algorithm can draw a sample according to the distribution $\cD_\ell$ in polynomial time and can compute the target $\Gamma_\ell$ on all the points of the support of the distribution in polynomial time.
\end{proof}

We now prove
\begin{theorem}\label{Th2}
Assuming randomized ETH, there is a constant $c$ such that any PAC learning algorithm for $n$-variable \mbox{\textsc{{\it size}-$s$ Monotone DT}  {\it and size}-$s$ \textsc{Monotone DNF}} by {\textsc {DNF}} with $\epsilon=1/(16n)$ must take at least 
$$n^{c\frac{\log\log s}{\log\log\log s}}$$ time.

The lower bound holds, even if the learner knows the distribution, can draw a sample according to the distribution in polynomial time and can compute the target on all the points of the support of the distribution in polynomial time.
\end{theorem} 
\begin{proof}
By Lemma~\ref{Th01}, assuming randomized ETH, there is a constant $c$ such that any PAC learning algorithm for $n$-variable {{\textsc {Monotone $(\log n)$-Junta}}} by {\textsc {DNF}} with $\epsilon=1/(16n)$ runs in time
$$n^{c\frac{\log\log n}{\log\log\log n}}.$$ Now by (\ref{Mon}) and since $s=n$, the result follows.
\end{proof}

\section{Tight Bound Assuming some Conjecture}\label{LLast}
A plausible conjecture on the hardness of \SC is the following.
\begin{conjecture}\label{conj1}\cite{KochST} 
There are constants $\alpha,\beta,\lambda\in(0,1)$ such that, for $k<N^\alpha$, there is no randomized $N^{\lambda k}$ time algorithm that can solve $\left(k,(1-\beta)\cdot k\ln N\right)$-\SC on $N$ vertices with high probability. 
\end{conjecture}
We now prove
\begin{theorem}\label{TTh2}
Assuming Conjecture~\ref{conj1}, there is a constant $c$ such that any PAC learning algorithm for $n$-variable \textsc{Monotone} $(\log s)$-\textsc{Junta}, \mbox{\textsc{{\it size}-$s$ Monotone DT}  {\it and size}-$s$ \textsc{Monotone DNF}} by {\textsc {DNF}} with $\epsilon=1/(16n)$ must take at least 
$$n^{c\log s}$$ time.

The lower bound holds, even if the learner knows the distribution, can draw a sample according to the distribution in polynomial time and can compute the target on all the points of the support of the distribution in polynomial time.
\end{theorem} 
\begin{proof} We give the proof for \textsc{Monotone} $(\log s)$-\textsc{Junta}. As in the proof of Theorem~\ref{Th2}, the result then follows for the other classes.

Consider the constants $\alpha,\beta$ and $\lambda$ in Conjecture~\ref{conj1}. Let $c=\min(\lambda/10,(1-\beta)/(20\log e))$. Suppose there is a PAC learning algorithm $\cA$ for {{\textsc {Monotone $(\log s)$-Junta}}} by {\textsc {DNF}} with $\epsilon=1/(16n)$ that runs in time $n^{c\log s}$. 
We show that there is $k<N^\alpha$, $k=\omega(1)$, such that $(k,k')$-\SC
can be solved in time $N^{\lambda k}$ where $k'=(1-\beta)k\ln N$. By Conjecture~\ref{conj1}, the result then follows.

Consider the following algorithm $\cB$
\begin{enumerate}
\item Input $\cS=(S,U,E)$ an instance for $(k,k')$-\SC.
\item Construct $\Gamma_5$ and $\cD_5$.
\item Run $\cA$ using $\Gamma_5$ and $\cD_5$ with $s=2^{5 k}$.  If it runs more than $N^{5c k}$ steps, then output \No.
\item Let $F$ be the output DNF.
\item Estimate $\eta=\dist_{\cD_5}(F,\Gamma_5)$.
\item If $\eta\le \frac{1}{16N}$, output \Yes, otherwise output \No.
\end{enumerate}
Since $c<\lambda/10$, the running time of this algorithm is $N^{5ck}<N^{\lambda k}$. Therefore, it is enough to prove the following
\begin{claim}
Algorithm $\cB$ solves $(k,k')$-\SC.
\end{claim}
\begin{proof}
\Yes case: Let $\cS=(S,U,E)$ be a $(k,k')$-\SC instance and $\opt(\cS)\le k$. Then, $5\cdot \opt(\cS)\le 5k=\log s$, and by Fact~\ref{oell}, $\Gamma_5$ is \textsc{Monotone $\log s$-Junta}. Therefore, w.h.p., algorithm $\cA$ learns $\Gamma_5$ and outputs a DNF that is $\eta=1/(16N)$ close to the target with respect to~$\cD_5$. Since $\cB$ terminates $\cA$ after $N^{5c k}$ time, we only need to prove that $\cA$ runs at most $N^{5c k}$ time. 

The running time of $\cA$ is
$$n^{c \log s}\le N^{5c  k}.$$

No Case: Let $\cS=(S,U,E)$ be a $(k,k')$-\SC instance and $\opt(\cS)> k'=(1-\beta)k\ln N$. By Lemma~\ref{error}, any DNF, $F$, of size $|F|<2^{k'/4}$ satisfies $\dist_{\cD_5}(F,\Gamma_5)\ge 1/(8|U|)-2^{-k'/4}$. 
Since, $c<(1-\beta)/(20\log e)$, 
$$2^{k'/4}=2^{\frac{(1-\beta)k\ln N}{4}}=N^{\frac{(1-\beta)k}{4\log e}}>N^{5ck},$$
any DNF, $F$, that the learning outputs satisfies $$\dist_{\cD_5}(F,\Gamma_5)\ge \frac{1}{8|U|}-2^{-k'/4}\ge \frac{1}{8N}-\frac{1}{N^{5ck}}\ge \frac{1}{9N}.$$
Therefore, with high probability the algorithm answer \No. 
\end{proof}
\end{proof}

\section{Strictly Proper Learning}

In this section, we prove 
\begin{theorem}\label{THEND}
Assuming randomized ETH, there is a constant $c$ such that any PAC learning algorithm for $n$-variable \textsc{Monotone} $(\log s)$-\textsc{Junta}, \mbox{\textsc{{\it size}-$s$ Monotone DT}  {\it and size}-$s$ \textsc{Monotone DNF}} by {\it size}-$s$ {\textsc {DNF}} with $\epsilon=1/(16n)$ must take at least 
$$n^{c\log s}$$ time.

The lower bound holds, even if the learner knows the distribution, can draw a sample according to the distribution in polynomial time and compute the target on all the points of the support of the distribution in polynomial time.
\end{theorem}

We first prove the following stronger version of Lemma~\ref{error}
\begin{lemma}\label{errorA}
Let $\cS=(S,U,E)$ be a set cover instance, and let $\ell\ge 5$. If $F:(\{0,1\}^\ell)^n\to \{0,1\}$ is a DNF of size $|F|<2^{\opt(\cS)\ell/16}$, then $\dist_{\cD_\ell}(F,\Gamma_\ell)\ge 1/(8|U|)$.
\end{lemma}

To prove this lemma, we will give some more results. 

Recall that, for a term $T$, $T_\cM$ is the conjunction of all the unnegated variables in $T$. We define the {\it monotone size} of $T$ to be $|T_\cM|$. For a DNF $F=T_1\vee T_2\vee \cdots \vee T_s$ and $z\in (\{0,1\}^\ell)^n$, we define the {\it monotone width} of $z$ in $F$ as
$$\mwidth_F(z):=\left\{
\begin{array}{ll}
\min_{T_i(z)=1}|(T_i)_\cM|&F(z)=1\\
0&F(z)=0
\end{array}.
\right. $$
We define $F^{-1}(1)=\{z|F(z)=1\}$ and $$\Omega=\Delta^1_n\cap F^{-1}(1).$$
\begin{claim}\label{coorA}
Let $F$ be a DNF with $$\dE{\bz\sim \cU(\Omega)}\left[\mwidth_F(\bz)\right]\le \opt(S)\cdot\ell/4.$$ Then 
$$\Prr{\bz\sim\cU(\Omega),\bj\sim \cU(\one(\bz)),\by\sim\cU(\bz^{\bj\gets U})} [F(\by)=1]\ge \frac{1}{2|U|}.$$
\end{claim}
\begin{proof} Let $z\in\Omega$. Then $F(z)=1$ and $z\in\Delta^1_n$. Let $T^z$ be the term in $F$ with $|T^z_\cM|=\mwidth_F(z)$ that satisfies $T^z(z)=1$. Let $Y_0=\{y_{i,m}|z_{i,m}=0\}$ and $Y_1=\{y_{i,m}|z_{i,m}=1\}$. Since $T^z(z)=1$, every variable in $Y_0$ that appears in $T^z$ must be negated, and every variable in $Y_1$ that appears in $T^z$ must be unnegated. 
For $j\in \one(z)$, define $q_z(j)$ to be the number of variables in $\{y_{1,j_1},\ldots,y_{n,j_n}\}$ that appear in $T^z(y)$. All those variables appear unnegated in $T$ because $j\in \one(z)$. Each variable in $T^z_\cM$ contributes $\lceil\ell/2\rceil^{n-1}$ to the sum $\sum_{j\in one(z)}q_z(j)$ and $|\one(z)|=\lceil\ell/2\rceil^n$. Therefore,
$$\dE{\bj\sim \cU(\one(z))}[q_z(\bj)]=\frac{|T^z_\cM|}{\lceil \ell/2\rceil}=\frac{\mwidth_F(z)}{\lceil \ell/2\rceil}.$$
Now, 
\begin{eqnarray*}
\dE{\bz\sim\cU(\Omega),\bj\sim\cU(\one(\bz))}[q_{\bz}(\bj)] 
&=& 
\frac{\dE{\bz\sim \cU(\Omega)}\left[\mwidth_F(\bz)\right]}{\lceil\ell/2\rceil}\\
&\le& \frac{\opt(\cS)}{2}.
\end{eqnarray*}
By Markov's bound, 
$$\Prr{\bz\sim\cU(\Omega),\bj\sim\cU(\one(\bz))}[q_{\bz}(\bj)<\opt(\cS)]\ge \frac{1}{2}.$$

Suppose for some $z\in\Omega$ and $j\in \one(z)$, we have $q_z(j)< \opt(\cS)$. Let $T^j$ be the conjunction of all the variables that appear in $T_\cM^z$ of the form $y_{i,j_i}$. Then $|T^j|=q_z(j)< \opt(\cS)$. 
By Fact~\ref{factopt}, there is $u\in U$ such that $T^j(u)=1$.
By Fact~\ref{Factz}, we have  $T^z(z^{j\gets u})=T^j(u)=1$. Then $F(z^{j\gets u})=1$. Since by item~\ref{Factz1} in Fact~\ref{Factz}, $|z^{j\gets U}|=|U|$, we have $$\Prr{\bz\sim\cU(\Omega),\bj\sim\cU(one(\bz)),\by\sim\cU(\bz^{\bj\gets U})}[F(\by)=1|q_{\bz}(\bj)<\opt(\cS)]\ge \frac{1}{|U|}.$$ 
Therefore,
\begin{eqnarray*}
\Prr{\bz\sim\cU(\Omega),\bj\sim \cU(\one(\bz)),\by\sim\cU(\bz^{\bj\gets U})}[F(\by)=1]&\ge& \Prr{\bz\sim\cU(\Omega),\bj\sim\cU(\one(\bz))}[q_{\bz}(\bj)<\opt(\cS)]\cdot\\ &&\Prr{\bz\sim\cU(\Omega),\bj\sim\cU(one(\bz)),\by\sim\cU(\bz^{\bj\gets U})}[F(\by)=1|q_{\bz}(\bj)<\opt(\cS)]\\
&\ge&\frac{1}{2|U|}.
\end{eqnarray*}

\end{proof}

\begin{claim}\label{pop}
Let $\cS=(S,U,E)$ be a set cover instance, and let $\ell\ge 5$. If $F:(\{0,1\}^\ell)^n\to \{0,1\}$ is a DNF and $\dE{\bz\sim\cU(\Omega)}[\mwidth_F(\bz)]\le \opt(\cS)\ell/4$, then $\dist_{\cD_\ell}(F,\Gamma_\ell)\ge 1/(8|U|)$.
\end{claim}
\begin{proof} 
If $\Prr{\by\sim\cU(\Delta_n^1)}[F(\by)\not=1]\ge 1/(4|U|)$, then by Fact~\ref{halff}, we have 
$$\dist_{\cD_\ell}(\Gamma_\ell,F)\ge \Prr{\by\sim\cD_\ell}[\Gamma_\ell(\by)\not=F(\by)|\Gamma_\ell(\by)=1]\Prr{\by\sim \cD_\ell}[\Gamma_\ell(\by)=1]=\frac{1}{2}\Prr{\by\sim \cU(\Delta^1_n)}[F(\by)\not=1]\ge \frac{1}{8|U|}.$$
If $\Prr{\by\sim\cU(\Delta_n^1)}[F(\by)\not=1]< 1/(4|U|)$, then by Fact~\ref{halff} and~\ref{sample}, and Claim~\ref{coorA}, 
\begin{eqnarray*}
\dist_{\cD_\ell}(\Gamma_\ell,F)&\ge& \Prr{\by\sim\cD_\ell}[\Gamma_\ell(\by)\not=F(\by)|\Gamma_\ell(\by)=0]\Prr{\by\sim\cD_\ell}[\Gamma_\ell(\by)=0]\\
&=&\frac{1}{2}\Prr{\by\sim\cU(\Delta_n^0)}[F(\by)=1]\\
&=&\frac{1}{2}\Prr{\bz\sim\cU(\Delta^1_n), \bj\sim\cU(\one(\bz)),\by\sim \cU(\bz^{\bj\gets U})}[F(\by)=1]\\
&\ge&\frac{1}{2}\Prr{\bz\sim\cU(\Delta^1_n), \bj\sim\cU(\one(\bz)),\by\sim \cU(\bz^{\bj\gets U})}[F(\by)=1|F(\bz)=1]\cdot \Prr{\bz\sim\cU(\Delta^1_n)}[F(\bz)=1]\\
&=&\frac{1}{2}\Prr{\bz\sim\cU(\Omega), \bj\sim\cU(\one(\bz)),\by\sim \cU(\bz^{\bj\gets U})}[F(\by)=1]\cdot \Prr{\bz\sim\cU(\Delta^1_n)}[F(\bz)=1]\\
&\ge&\frac{1}{2}\frac{1}{2|U|}\left(1-\frac{1}{4|U|}\right)\ge \frac{1}{8|U|}.
\end{eqnarray*}
\end{proof}

\begin{claim}\label{jkl}
Let $F$ be a size-$s$ DNF formula for $s\ge 2$ such that $\dist_{\cD_\ell}(F,\Gamma_\ell)\le 1/4$, then
$$\dE{\by\sim\cU(\Omega)}[\mwidth_F(\by)]\le 4\log s.$$
\end{claim}
\begin{proof}
First, we have
\begin{eqnarray*}
\frac{3}{4}&\le& \Prr{\by\sim \cD_\ell}[F(\by)=\Gamma_\ell(\by)]\\&=& \frac{1}{2}\Prr{\by\sim\cD_\ell}[F(\by)=\Gamma_\ell(\by)|\Gamma_\ell(\by)=1]+\frac{1}{2}\Prr{\by\sim\cD_\ell}[F(\by)=\Gamma_\ell(\by)|\Gamma_\ell(\by)=0]\\
&\le& \frac{1}{2}\Prr{\by\sim\cU(\Delta_n^1)}[F(\by)=1]+\frac{1}{2}.
\end{eqnarray*}
 Therefore, $\Pr_{\by\sim\cU(\Delta^1_n)}[F(\by)=1]\ge 1/2.$ 
 
Let $F=T_1\vee T_2\vee\cdots \vee T_s$. For $y\in \Omega$, let $\omega(y)\in [s]$ be the minimum integer such that $\mwidth_F(y)=|(T_{\omega(y)})_\cM|$ and $T_{\omega(y)}(y)=1$.
 
 Then, by (\ref{lstep}), 
$$\Prr{\by\sim\cU(\Omega)}[T_i(\by)=1]=\Prr{\by\sim\cU(\Delta^1_n)}[T_i(\by)=1|F(\by)=1]=\frac{\Prr{\by\sim\cU(\Delta^1_n)}[T_i(\by)=1]}{\Prr{\by\sim\cU(\Delta^1_n)}[F(\by)=1]}\le 2^{-|(T_i)_\cM|/2+1}.$$
Now, by the concavity of $\log$,
\begin{eqnarray}
\frac{1}{2}\dE{\by\sim\cU(\Omega)}[\mwidth_F(\by)]-1&=& \dE{\by\sim\cU(\Omega)}\left[\log \left(2^{\mwidth_F(\by)/2-1}\right)\right]\nonumber\\
&\le& \log\left(\dE{\by\sim\cU(\Omega)}\left[ 2^{\mwidth_F(\by)/2-1}\right]\right)\nonumber\\
&= & \log\left(\sum_{i\in [s]} 2^{|(T_i)_\cM|/2-1}\Prr{\by\sim\cU(\Omega)}[\omega(\by)=i]\right)\nonumber\\
&\le & \log\left(\sum_{i\in [s]} 2^{|(T_i)_\cM|/2-1}\Prr{\by\sim\cU(\Omega)}[T_i(\by)=1]\right)\nonumber\\
&\le&\log\left(\sum_{i\in [s]} 2^{|(T_i)_\cM|/2-1}2^{-|(T_i)_\cM|/2+1}\right)\nonumber\\
&=&\log s.\nonumber
\end{eqnarray}
Therefore, $\dE{\by\sim\cU(\Omega)}[\mwidth_F(\by)]\le 4\log s.$
\end{proof}

We are now ready to prove Lemma~\ref{errorA}
\begin{proof}
If $\dist_{\cD_\ell}(F,\Gamma_\ell)>1/4$, then the result follows. Now suppose $\dist_{\cD_\ell}(F,\Gamma_\ell)\le 1/4$. If $s=|F|<2^{\opt(\cS)\ell/16}$, then by Claim~\ref{jkl}, $\dE{\by\sim\cU(\Omega)}[\mwidth_F(\by)]\le 4\log s=\opt(\cS)\ell/4$. Then by Claim~\ref{pop}, $\dist_{\cD_\ell}(F,\Gamma_\ell)\ge1/(8|U|).$  
\end{proof}

The proof of  Theorem~\ref{THEND} is the same as the proof of Theorem~14 in \cite{KochST}. We give the proof for completeness.
\begin{proof}
Consider the constant $\lambda$ in Lemma~\ref{Lin}. Let $c=\lambda/6$.
Suppose there is a PAC learning algorithm $\cA$ for {{\textsc {Monotone $(\log s)$-Junta}}} by size-$s$ {\textsc {DNF}} with $\epsilon=1/(16n)$ that runs in time $n^{c{\log s}}$. We show that there is $k$ such that for
$$k'=\frac{1}{2}\left(\frac{\log N}{\log\log N}\right)^{1/k},$$ $(k,k')$-\SC
can be solved in time $N^{5ck}\le N^{\lambda k}$. By Lemma~\ref{Lin}, the result then follows.

Let $\cS=(S,U,E)$ be an $N$-vertex $(k,k')$-\SC instance where 
$$k=\frac{1}{2}\frac{\log\log N}{\log\log\log N} \mbox{\ and\ } k'=\frac{1}{2}\left(\frac{\log N}{\log\log N}\right)^{1/k}.$$

Consider the following algorithm $\cB$
\begin{enumerate}
\item Input $\cS=(S,U,E)$ an instance for $(k,k')$-\SC.
\item Construct $\Gamma_5$ and $\cD_5$.
\item Run $\cA$ using $\Gamma_5$ and $\cD_5$ with $s=2^{5k}$ and $n=N$.  If it runs more than $N^{5c k}$ steps, then output \No.
\item Let $F$ be the output DNF.
\item If $|F|>s$ then output \No.
\item Estimate $\eta=\dist_{\cD_5}(F,\Gamma_5)$.
\item If $\eta\le \frac{1}{16N}$, output \Yes, otherwise output \No.
\end{enumerate}
The running time of this algorithm is $N^{5c k}\le N^{\lambda k}$. Therefore, it is enough to prove the following
\begin{claim}
Algorithm $\cB$ solves $(k,k')$-\SC.
\end{claim}
\begin{proof}
\Yes case: Let $\cS=(S,U,E)$ be a $(k,k')$-\SC instance and $\opt(\cS)\le k$. Then, $size(\Gamma_5)\le 2^{5\cdot\opt(\cS)}\le 2^{5k}=s$, and by Fact~\ref{oell}, $\Gamma_5$ is \textsc{Monotone $\log s$-Junta}. Therefore, w.h.p., algorithm $\cA$ learns $\Gamma_5$ and outputs a DNF that is $\eta=1/(16N)$ close to the target with respect to $\cD_5$. Since $\cB$ terminates $\cA$ after $N^{5ck}$ time, we only need to prove that $\cA$ runs at most $N^{5c k}$ time.

The running time of $\cA$ is
$$n^{c\log s}=N^{c\log s}\le N^{5ck}.$$

\No Case: Let $\cS=(S,U,E)$ be a $(k,k')$-\SC instance and $\opt(\cS)> k'$. By Lemma~\ref{errorA}, any DNF, $F$, of size $|F|<2^{5k'/16}$ satisfies $\dist_{\cD_5}(F,\Gamma_5)\ge 1/(8|U|)$.
First, we have, for large $N$
$$k'=\frac{1}{2}\left(\frac{\log N}{\log\log N}\right)^{1/k}>32k.$$
Therefore, any DNF, F, of size $|F|<2^{10k}$ satisfies $\dist_{\cD_5}(F,\Gamma_5)\ge 1/(8|U|)$.

We have $2^{10k}>s$. So, $\cB$ either runs more than $N^{5ck}$ steps and then outputs \No in step 3 or outputs a DNF of size more than $s$ and then outputs \No in step~4 or outputs a DNF of size at most $s$ with $\dist_{\cD_5}(F,\Gamma_5)\ge 1/(8|U|)>1/(8N)>1/(16N)$ and outputs \No in step 6.
\end{proof}
\end{proof}

\bibliography{TestingRef}

\appendix

\end{document}